\documentclass[conference]{IEEEtran}

\usepackage{epsfig}
\usepackage{amsmath}
\usepackage{amsfonts}
\usepackage{amssymb}
\usepackage{subfigure}
\usepackage{booktabs}

\newtheorem{theorem}{Theorem}

\newtheorem{corollary}{Corollary}

\DeclareMathOperator{\ord}{o}

\DeclareMathOperator*{\argmax}{arg\,max}

\newcommand{\Igmi}{R_{\rm gmi}}
\newcommand{\Igmis}{R_{\rm gmi}(s)}
\newcommand{\Cmlc}[1]{{\sf C}^{\rm mlc}_{ #1}}

\newcommand{\Xc}{\mathcal{X}}
\newcommand{\Yc}{\mathcal{Y}}
\newcommand{\Mc}{\mathcal{M}}

\newcommand{\Xjb}{\Xc_b^j}
\newcommand{\Cc}{{\cal C}}

\newcommand{\bv}{{\boldsymbol b}}
\newcommand{\xv}{{\boldsymbol x}}
\newcommand{\yv}{{\boldsymbol y}}

\newcommand{\Xm}{{\boldsymbol X}}
\newcommand{\Ym}{{\boldsymbol Y}}

\newcommand{\msf}{{\sf m}}

\newcommand{\eqdef}{\stackrel{\Delta}{=}}

\newcommand{\beq}{\begin{equation}}
\newcommand{\eeq}{\end{equation}}

\newcommand{\snr}{{\sf snr}}

\newcommand{\EbNo}{\frac{E_{\rm b}}{N_0}}

\newcommand{\ebnolim}{\frac{E_{\rm b}}{N_0}_{\rm lim}}

\newcommand{\e}{{\rm e}}

\newcommand{\EE}{\mathbb{E}}
\newcommand{\FF}{\mathbb{F}}
\newcommand{\CC}{\mathbb{C}}

\newcommand{\pd}{P}

\newcommand{\Ccm}[1]{{\sf C}^{\rm cm}_{ #1}}
\newcommand{\Cbicm}[1]{{\sf C}^{\rm bicm}_{ #1}}

\newcommand{\ca}{c_1}
\newcommand{\cb}{c_2}

\title{Shaping Bits}

\author{
\authorblockN{Albert Guill{\'e}n i F{\`a}bregas}
\authorblockA{Engineering Department \\
University of Cambridge, UK \\
\tt guillen@ieee.org}
\and
\authorblockN{Alfonso Martinez}
\authorblockA{Centrum Wiskunde \& Informatica\\
Amsterdam, The Netherlands\\
\tt alfonso.martinez@ieee.org}
}

\begin{document}

\maketitle

\begin{abstract}
The performance of bit-interleaved coded modulation (BICM) with bit shaping (i.e., non-equiprobable bit probabilities in the underlying binary code) is studied. For the Gaussian channel, the rates achievable with BICM and bit shaping are practically identical to those of coded modulation or multilevel coding. This identity holds for the whole range of values of signal-to-noise ratio. Moreover, the random coding error exponent of BICM significantly exceeds that of multilevel coding and is very close to that of coded modulation. 
\end{abstract}

\footnotetext[1]{This work has been supported by the International Joint Project 2008/R2 of the Royal Society.}
\setcounter{footnote}{1}

\section{Introduction}

For non-binary modulations in the Gaussian channel, three main constructions for coding schemes achieve information rates close to the channel capacity are known: coded modulation (CM), bit-interleaved coded modulation (BICM), and multilevel coding (MLC). CM dates back to the pioneering work of Ungerb{\"o}ck \cite{ungerboeck1982ccm}, and merges coding and modulation in a single entity. In contrast, BICM separates them, and is built around the mapping of a simple binary code onto a non-binary modulation \cite{zehavi1992ptc,caire1998bic,guillen2008bic}. MLC makes use of a layer of binary codes, one for each bit in the binary label of the modulation symbol \cite{imai1977nmc,wachsmann1999mct}. 

CM allows for the highest information rates. It is closely followed by multilevel coding (for equiprobable modulation symbols the rates coincide) and, with a larger loss, by BICM. In terms of error exponents, the situation is somewhat reversed, with CM again the best, but now BICM beats multilevel coding at low rates. Whereas previous analyses in the literature assume that the modulation symbols are used equiprobably, in this work we lift this assumption and consider shaping, whereby the bit or symbol probabilities are arbitrary. We will see that BICM with shaping achieves both information rates and error exponents very close to those of CM, thus closing the gap which made multilevel coding better in terms of information rates.

This paper is organized as follows. In Sect.~\ref{sec:preliminaries} we introduce the necessary concepts and notation describing the various schemes. In Sect.~\ref{sec:achievableRates} we derive the achievable rates for BICM with shaping by using mismatched decoding theory. These results are particularized for the Gaussian channel in Sect.~\ref{sec:shapingGaussian}, which also includes some numerical results.


\section{Preliminaries}
\label{sec:preliminaries}
\subsection{Blockwise Coded Transmission}

Consider a memoryless channel with input $X$ and output $Y$, respectively belonging to the sets $\Xc$ and $\Yc$. A block code $\Mc\subseteq\Xc^N$ is a set of $|\Mc|$ vectors (or codewords) $\xv$ of length $N$ (the number of channel uses), i.\ e.\ $\xv =(x_1,\dotsc,x_N)\in \Xc^N$. The output $\yv\eqdef(y_1,\dotsc,y_N)$ is a random transformation of the input with transition probability distribution $P_{\Ym|\Xm}(\yv|\xv)$. 
For memoryless channels the distribution $P_{\Ym|\Xm}(\yv|\xv)$ admits the decomposition
\beq
P_{\Ym|\Xm}(\yv|\xv) = \prod_{k=1}^N P_{Y|X}(y_k|x_k)
\eeq
With no loss of generality, we limit our attention to continuous output and identify $P_{Y|X}(y|x)$ as a probability density function. We adopt the convention that capital letters represent random variables, while the corresponding small letters correspond to realizations of the variables.


At the source, a message $\msf$ drawn with equal probability from a message set $\Mc$
is mapped onto a codeword $\xv$.
We denote this encoding function by $\phi$, i.\ e.\ $\phi(\msf) = \xv$. 
The corresponding transmission rate $R$ is given by $R \eqdef \frac{1}{N}\log|\Mc|$.
At the receiver, the decoder determines the codeword {\em decoding metric}, denoted by $q(\xv,\yv)$, for all codewords, and outputs the message $\widehat{\msf}$ whose  metric is largest,
\begin{align}
\widehat{\msf}& = \argmax_{\msf\in\{1,\dotsc,|\Mc|\}} q(\phi(\msf),\yv) 
\label{eq:min_metric_dec}.
\end{align}
The metrics we consider are products of symbol decoding metrics $q(x,y)$, namely (with some abuse of notation)
\beq
q(\xv,\yv) = \prod_{k=1}^Nq(x_k,y_k).
\label{eq:symbol_metric}
\eeq 

For maximum likelihood (ML) decoders, the decoding metric is given by $q(x,y) = P_{Y|X}(y|x)$. More generally, a decoder finds the most likely codeword as long as the metric $q(x,y)$ is a strictly increasing bijective function of the channel transition probability $P_{Y|X}(y|x)$. Instead, if the metric $q(x,y)$ is not a bijective function of the channel transition probability, 
we have a {\em mismatched decoder} \cite{merhav1994irm,ganti2000mdr}. 

Of special interest is the random ensemble corresponding to CM, for which 1) the channel inputs are selected independently for each codeword component according to a probability distribution $P_{X}(x)$, and 2) the decoder uses the ML metric. In this case, and for practical reasons, the modulation set $\Xc$ is taken finite. Let $M\eqdef|\Xc|$ denote the cardinality of $\Xc$ and $m \eqdef \log_2M$ the number of bits 
required to index a symbol.
%
The largest information rate that can be achieved with CM under the constraint $x\in\Xc$ is  is 
\beq\label{eq:opt_cm}
\Ccm{} =\sup_{P_{X}(X)}I(X;Y).
\eeq
Moreover, for any input distribution $P_X(X)$, the block error probability $P_e$ satisfies \cite{gallager1968ita}
\begin{align}
\bar{P}_e &\leq \e^{-NE_{\rm r}(R)}
\label{eq:pe_ave_err_exp}
\end{align}
where $E_{\rm r}(R)=\sup_{0\leq\rho\leq 1}E_0(\rho) - \rho R$ and
\beq
E_0(\rho) \eqdef -\log \EE \left[\left(\sum_{x'} P_X(x') \left(\frac{P_{Y|X}(Y|x')}{P_{Y|X}(Y|X)}\right)^{\frac{1}{1+\rho}}\right)^\rho\right].
\label{eq:generalized_gallager}
\eeq
The expectation is carried out according to the joint distribution $P_{X,Y}(x,y)=\pd_{Y|X}(y|x)P_X(x)$.

\subsection{Bit-Interleaved Coded Modulation}

In practical CM schemes, since the codewords are selected elements of $\Xc^N$ and the alphabet $\Xc$ has typically more than 2 elements, the corresponding codes are in some sense non-binary. BICM is a different construction where the underlying code is binary. Originally analyzed in \cite{caire1998bic} under the assumption of infinite-depth interleaving, this restriction was recently lifted in \cite{guillen2008bic,martinez2009bic}, where it was shown that BICM has a natural description in terms of mismatched decoding.

The BICM encoder generates a vector of $mN$ bits, $\bv = (b_1,\dotsc,b_{mN})$, i.\ e.\ $\phi(\msf) = \bv$. This vector is mapped onto a vector of $N$ modulation symbols according to a labeling rule $\mu:\FF_2^m\to\Xc$, such that
\beq
x_k = \mu\bigl(b_{(k-1)m+1},\dotsc,b_{(k-1)m+1}\bigr), \quad k = 1,\dotsc,N.
\eeq 
Note that the interleaver which gives its name to BICM has been absorbed in this description of the encoder. Analogously, we denote the inverse labeling by $b_j$, so that $b_j(x)$ is the $j$-th bit in the binary label of modulation symbol $x$, for $j = 1, \dotsc, m$. By construction, the modulation symbols $x$ are used with probabilities 
\beq
P_X^{\rm bicm}(x) = \prod_{j=1}^m P_{B_j}\bigl(b_j(x)\bigr).
\label{eq:px_bicm}
\eeq

In addition to the different code construction, BICM also differs from CM at the receiver side. First, let us define the sets $\Xjb$ as those elements of $\Xc$ having bit $b$ in the $j$-th label position, i.e., $\Xjb \eqdef \{ x\in\Xc : b_j(x) = b\}$. 
The BICM symbol metric combines the $m$ bit metrics $q_j(b_j,y)$ given by
\beq
q_j(b_j(x) = b,y) = \sum_{x'\in\Xjb}P_{Y|X}(y|x') P_X^{\rm bicm}(x'),
\label{eq:bit_metric_sum}
\eeq
as if the $m$ bits in a symbol were independent, i.e.,
\beq
q^{\rm bicm}(x,y) = \prod_{j=1}^m q_j\bigl(b_j(x),y\bigr).
\eeq
Hence, the BICM receiver uses the following symbol metric
\beq
q(x,y) = \prod_{j=1}^m \Biggl(\sum_{x'\in\Xc^j_{b_j(x)}}
P_{Y|X}(y|x')\prod_{j'=1}^m P_{B_{j'}}(b_{j'}(x'))\Biggr).
\label{eq:symbol_metric_sum}
\eeq

\subsection{Multilevel Coding}

Multilevel codes (MLC) combined with multistage decoding (MSD) have been proposed \cite{imai1977nmc,wachsmann1999mct} as an efficient method to attain the channel capacity by using binary codes. For BICM, a single binary code $\Cc$ is used to generate a binary codeword, which is used to select modulation symbols by a binary labeling function $\mu$. 
In MLC, the input binary code $\Cc$ is the Cartesian product of $m$ binary codes of length $N$, one per modulation level, i.\ e.\ $\Cc = \Cc_1 \times \dotsc \times \Cc_m$, and the input distribution for the symbol $x(b_1,\dotsc,b_j)$ has the form
\beq
P_X^{\rm mlc}(x) = P_{B_1,\dotsc,B_M}(b_1,\dotsc,b_m) = \prod_{j=1}^m P_{B_j}(b_j).
\label{eq:symbol_prob_mlc}
\eeq
For a fixed input distribution on the bits, MLC achieve the mutual information \cite{imai1977nmc,wachsmann1999mct} both with ML joint decoding and with multistage decoding. The largest information rate that can be achieved with MLC under the constraint $x\in\Xc$ is 
\beq\label{eq:opt_mlc}
\Cmlc{} =\sup_{\substack{P_{B_1}(B_1),\dotsc, P_{B_m}(B_m)}} I(X;Y).
\eeq
The error exponents of MLC with multistage decoding were derived in \cite{beyer2001aaa,beyer2001aaa2,guillen2008bic,ingber2009cee}, where it was also shown the error exponent is upper bounded by one.

\section{Achievable Rates with BICM}
\label{sec:achievableRates}

For the BICM scheme described above, it was shown in \cite{guillen2008bic,martinez2009bic} that the rate
\beq
\Igmi = \sup_{s>0}\left(\EE\left[\log \frac{\prod_{j=1}^m q_j\bigl(b_j(X),Y\bigr)^s}{\frac{1}{M}\sum_{x'}\prod_{j=1}^m q_j\bigl(b_j(x'),Y\bigr)^s }\right]\right),
\eeq
also named generalized mutual information (GMI), is achievable with equiprobable bits, $P_{B_j}(b) = \frac{1}{2}$. The proof is based on a simple extension of Gallager's analysis of ML decoding in terms of error exponents to mismatched decoding \cite{merhav1994irm,ganti2000mdr}. References \cite{guillen2008bic,martinez2009bic} also show that the above rate may be decomposed as the sum of $m$ bit GMI terms, and that it coincides with the BICM capacity defined in \cite{caire1998bic}. The next result generalizes this result for arbitrary bit probabilities.

\begin{theorem}\label{theorem:gmiBICM-sum}
The generalized mutual information of the BICM mismatched decoder is equal to the sum of the generalized mutual informations of $m$ binary-input channels,
\begin{align}\label{theorem:gmiBICM-sum-eq1}
\Igmi
&=\sup_{s>0}\left(\sum_{j=1}^m \EE\left[\log\frac{q_j(B_j,Y)^s}{\sum_{b'=0}^1q_j(b_j',Y)^s P_{B_j}(b'_j)}\right]\right).
\end{align}
The expectation is carried out according to the joint distribution $P_{B_j}(b_j)P_j(y|b_j)$, with 
\beq\label{theorem:gmiBICM-sum-eq3}
P_j(y|b) \eqdef \sum_{x\in\Xjb}\frac{P_{Y|X}(y|x)P_X^{\rm bicm}(x)}{\sum_{x'\in\Xjb}P_X^{\rm bicm}(x')}.
\eeq

An alternative expression is
\begin{align}\label{theorem:gmiBICM-sum-eq2}
\Igmi
& = \sup_{s>0}\left(\sum_{j=1}^m \EE\left[ \log\frac{q_j\bigl(b_j(X),Y\bigr)^s}{\sum_{b'=0}^1q_j(b_j',Y)^sP_{B_j}(b_j')}\right]\right),
\end{align}
where the expectation is done according to the joint distribution $P_X^{\rm bicm}(x)P_{Y|X}(y|x)$.
\end{theorem}

\begin{proof}
For fixed $s$ and probabilities $P_X^{\rm bicm}(x) = \prod_{j=1}^m P_{B_j}\bigl(b_j(x)\bigr)$ the GMI can be written as
\begin{align}
\Igmis
& = \EE\left[\log \frac{q^{\rm bicm}(X,Y)^s}{\sum_{x'}q^{\rm bicm}(x',Y)^s P_{X}^\text{bicm}(x')}\right] \\
& = \EE\left[\log \frac{\prod_{j=1}^m q_j\bigl(b_j(X),Y\bigr)^s}{\sum_{x'}\prod_{j=1}^m q_j\bigl(b_j(x'),Y\bigr)^s P_{B_j}\bigl(b_j(x')\bigr)}\right],
\label{eq:gmi_log}
\end{align}
where the expectation is carried out according to $P_X^{\rm bicm}(x)P_{Y|X}(y|x)$.

We now have a closer look at the denominator in the logarithm of \eqref{eq:gmi_log}. The key observation is that the sum over the constellation points ($x'\in\Xc$) of the product of a function $f\bigl(b_j(x')\bigr)$ evaluated at all the binary label positions admits an alternative expression, namely
\begin{align}
\sum_{x'\in\Xc}&\Biggl(\prod_{j=1}^m f\bigl(b_j(x')\bigr)\Biggr) = \prod_{j=1}^m \Biggl(\sum_{b_j\in \{0,1\}} f(b_j)\Biggr). 
\end{align}
Indeed, after carrying out the product in the right-hand side, we obtain the desired sum over all $2^m$ binary $m$-tuples $(b_1,\dotsc,b_m)$ of summands of the form $f(b_1)\cdots f(b_m)$.

Therefore, for the specific choice $f\bigl(b_j(x')\bigr) = q_j\bigl(b_j(x'),Y\bigr)^s P_{B_j}\bigl(b_j(x')\bigr)$ we have
the product over all label positions of the sum of the probabilities of the bit $b_j$  being zero and one, i.e.,
\begin{align}
\sum_{x'\in\Xc}&\left(\prod_{j=1}^m q_j\bigl(b_j(x'),Y\bigr)^s P_{B_j}(b_j(x))\right) \\
&= \prod_{j=1}^m \left(\sum_{b'\in\{0,1\}}q_j(b_j,Y)^sP_{B_j}(b_j')\right).
\end{align}

Next, going back to \eqref{eq:gmi_log}, we obtain
\begin{align}
\Igmis &= \EE \left[ \log\left(\prod_{j=1}^m\frac{q_j\bigl(b_j(X),Y\bigr)^s}{\sum_{b'=0}^1q_j(b_j',Y)^sP_{B_j}(b_j')}\right)\right], \\
&= \sum_{j=1}^m \EE\left[ \log\frac{q_j\bigl(b_j(X),Y\bigr)^s}{\sum_{b'=0}^1q_j(b_j',Y)^sP_{B_j}(b_j')}\right],
\end{align}
where the expectation is done according to the joint distribution $P_X^{\rm bicm}(x)P_{Y|X}(y|x)$. This gives Eq.~\eqref{theorem:gmiBICM-sum-eq2} since the generalized mutual information is the supremum over all $s$ \cite{merhav1994irm,ganti2000mdr}. As for Eq.~\eqref{theorem:gmiBICM-sum-eq2}, we derive it by noting that, for each $j$,  the summation over $x$ in the expectation can be split into two parts and rearranged as follows,
\begin{align}
\sum_x f(x) &= \sum_{b_j\in\{0,1\}}\sum_{x\in\Xjb}f(x) \\ &= \sum_{b_j\in\{0,1\}}P_{B_j}(b_j)\sum_{x\in\Xjb}\frac{f(x)}{P_{B_j}(b_j)}.
\end{align}
As $P_{B_j}(b_j) = \sum_{x'\in\Xjb}P_X^{\rm bicm}(x')$ by construction, recovering the expression of $f(x)$ we obtain $P_j(y|b_j)$ in Eq.~\eqref{theorem:gmiBICM-sum-eq3}.
\end{proof}

The following result applies to BICM with the decoding metric given in Eq.~\eqref{eq:bit_metric_sum}.
\begin{corollary}
For the classical BICM decoder with metric in Eq.~\eqref{eq:bit_metric_sum} the supremum over $s$ is achieved at $s=1$, and $\Igmi = \sum_{j=1}^m I(B_j;Y)$.
\end{corollary}
\begin{proof}
Since the metric $q_j(b_j,y)$ is proportional to $P_j(y|b_j)$, we can identify the quantity
\beq
\EE\left[\log\frac{q_j\bigl(B_j,Y\bigr)^s}{\sum_{b_j'=0}^1 q_j(b_j',Y)^s P_{B_j}(b_j') }\right]
\eeq
as the generalized mutual information of a {\em matched} binary-input channel with transitions $P_j(y|b_j)$. Then, the supremum over $s$ is achieved at $s=1$ (that is, the mutual information $I(B_j;Y)$) and we get the desired result.
\end{proof}
In the remainder of the paper, for the sake of simplicity and without loss of generality, we focus on the classical BICM metric. Clearly, the methods and results we present generalize to other metrics, in which case, $s$ should also be optimized.

The above results suggest that we can chose the input bit distribution that yields the largest GMI, i.e.,
effectively implying shaping the bit probabilities in BICM as
\beq\label{eq:opt_bicm}
\Cbicm{} =\sup_{\substack{P_{B_1}(B_1),\dotsc, P_{B_m}(B_m)}} \sum_{j=1}^m I(B_j;Y).
\eeq
For iid codebooks, $\Cbicm{}$ is also the largest rate that can be transmitted with vanishing error probability \cite{lapidoth1996nnd}.

This capacity should be compared with the equivalent quantities on CM and multi-level coding, given in Eqs.~\eqref{eq:opt_cm} and~\eqref{eq:opt_mlc} respectively,
\begin{align}
\Ccm{} &=\sup_{P_{X}(X)}I(X;Y), \\
\Cmlc{} &=\sup_{\substack{P_{B_1}(B_1),\dotsc, P_{B_m}(B_m)}} I(X;Y).
\end{align}
Note that BICM differs from CM in the transmitter, where the modulation symbol probabilities have the specific form $P_X^{\rm bicm}(x) = \prod_{j=1}^m P_{B_j}\bigl(b_j(x)\bigr)$, and the receiver, where the symbol metric in Eq.~\eqref{eq:symbol_metric_sum} is used for decoding.

In terms of the random coding error exponent, the analysis in \cite{guillen2008bic,martinez2009bic} can be merged with the previous proof to show that for any input distribution $P_X(x)$, the block error probability $P_e$ is upper bounded by
\begin{align}
\bar{P}_e &\leq \e^{-NE_{\rm r}^q(R)}
\label{eq:pe_ave_err_exp_bicm}
\end{align}
where $E_{\rm r}^q(R)=\sup_{\substack{0\leq\rho\leq1 \\s>0}}E_0^q(\rho,s) - \rho R$, and
\beq
E_0^q(\rho,s) \eqdef -\log \EE \left[\left(\sum_{x'} P_X(x') \left(\frac{q^{\rm bicm}(x',Y)}{q^{\rm bicm}(X,Y)}\right)^s\right)^\rho\right]
\label{eq:generalized_gallager_bicm}
\eeq
is a generalized Gallager function. The expectation is carried out according to the joint distribution $P_{Y|X}(y|x)P_X(x)$.


\section{Bit Shaping for the Gaussian Channel}
\label{sec:shapingGaussian}
\subsection{Channel Model}
We consider the transmission over complex-plane signal sets ($\Xc \subset \CC$, $\Yc = \CC$) in the AWGN channel. It is a memoryless channel satisfying
\beq
Y_k=\sqrt{\snr} \,X_k +Z_k,~~~~ k=1,\dotsc,N
\label{eq:model}
\eeq 
where $Z_k$ are zero-mean, unit-variance, circularly symmetric complex Gaussian samples, and $\snr$ is the signal-to-noise ratio (SNR). 
We wish to solve the optimization problems in Eqs.~\eqref{eq:opt_cm}, ~\eqref{eq:opt_mlc}  and ~\eqref{eq:opt_bicm} with the additional constraints that $x\in\Xc$, $\EE[X] = 0$, and $\EE[|X|^2] = 1$.

We consider binary reflected Gray mapping\footnote{Recall that the binary reflected Gray mapping for $m$ bits may be generated recursively from the mapping for $m-1$ bits by prefixing a binary 0 to the mapping for $m-1$ bits, then prefixing a binary 1 to the reflected (i.\ e.\ listed in reverse order) mapping for $m-1$ bits. For QAM modulations in the Gaussian channel, the symbol mapping is the Cartesian product of Gray mappings over the in-phase and quadrature PAM components.}. 
For shaping, $2^{m}$-QAM signal sets are of special interest; this constellation is the Cartesian product of two $2^\frac{m}{2}$-PAM constellations, one for each of the in-phase and quadrature components of the channel. Since the optimum input distribution is known to be Gaussian, a good input distribution over the set $\Xc$ should approach in some sense a Gaussian density. Symmetry between the in-phase and quadrature components and along the zero axis (so that the positive and negative plane have equal probability) dictate that the optimization problems in Eqs.~\eqref{eq:opt_cm} and~\eqref{eq:opt_mlc} respectively have 
\begin{itemize}
\item $2^{\frac{m}{2}-1}-1$ free parameters for CM, and
\item $\frac{m}{2}-1$ free parameters for BICM and MLC.
\end{itemize}
For BICM we used the symmetries of binary reflected Gray mapping and the fact that the most significant bit selects the positive or negative half-plane, and always has probability $\frac{1}{2}$.

Note that the CM optimization problem does not restrict the input distribution to be $P_X(x) = \prod_{j=1}^m P_{B_j}(b_j(x))$, hence being able to achieve potentially larger rates. As we shall see, the resulting difference in information rates is however marginal. Moreover, note that there is an exponential relationship between the number of free parameters for BICM and CM, which can induce rather large computational savings for large signal sets. 
For example, since for 16-QAM there is only one free parameter for MLC and CM, the optimization will result in the best performance, i.e., MLC is optimal and BICM, as we shall see, is very close. However, for $m>4$ this is no longer true and the optimization over symbol probabilities without restriction $P_X(x)$ to be the product of bit probabilities could potentially yield larger rates.

\subsection{Numerical Examples}

Figure \ref{fig:16qam_sh_equi} shows the improvement in BICM capacity derived from shaping for 16-QAM with binary reflected Gray mapping. As we observe $\Cbicm{}$ (dashed) is almost indistinguishable from $\Ccm{}$ or $\Cmlc{}$ (thin solid) or channel capacity itself (thick solid). This shows that shaping bits for BICM can recover the BICM capacity loss for equiprobable bits and effectively close the gap with CM and MLC.
Remark that the BICM demodulator is a one-shot non-iterative demodulator. 
In general, the decoding complexity of BICM is larger than that of MLC, since the codes of MLC are shorter. 
In practice, however, if the decoding complexity grows linearly with the number of bits in a a codeword, e.\ g.\ with LDPC or turbo codes, the overall complexity of BICM becomes comparable to that of MLC.
\begin{figure}[t]
 \centering 
 \includegraphics[width=1\columnwidth]{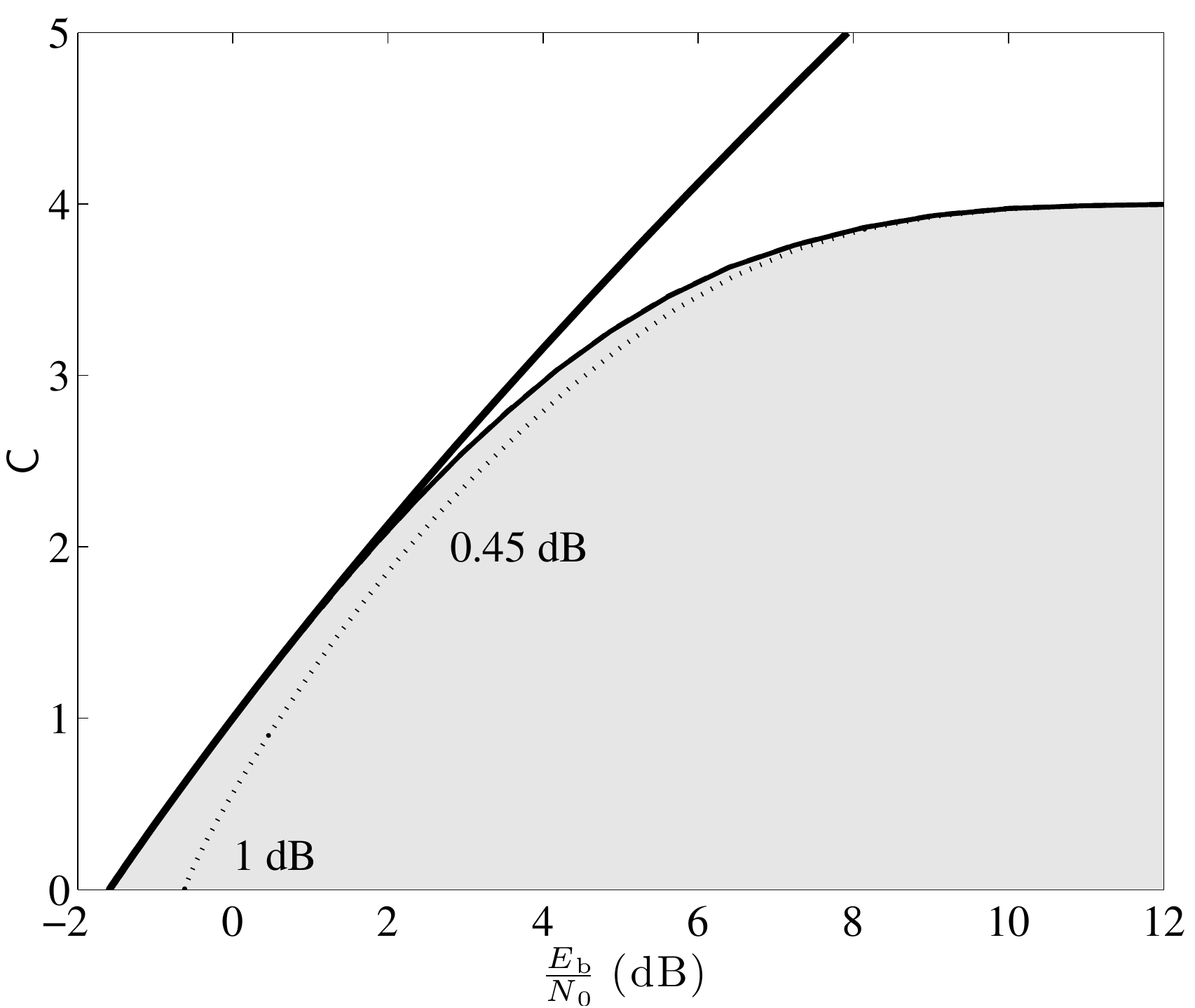}
 \caption{Capacities for Gaussian inputs (thick solid), CM/MLC with shaping (thin solid), BICM with shaping (dashed) and BICM with equiprobable inputs (dotted line) for 16-QAM with Gray mapping and bit shaping as a function of $\EbNo$ (dB).}
 \label{fig:16qam_sh_equi}
 \vspace{-4mm}
\end{figure}

Figure \ref{fig:exponents_cm_bicm} shows the error exponents for CM and BICM, with and without shaping, for 16-QAM at $\snr=8$ dB. When shaping is used, the input distribution is the corresponding optimal capacity achieving distribution. We observe that when shaping is used, in the region near capacity, the overall BICM error exponent is very close to that of CM, while when equiprobable bits are used, the exponent deviates from that of CM. Remark that, according to \cite{guillen2008bic,martinez2009bic} the BICM error exponent cannot be larger than that of CM,  as opposed to that of the independent parallel channel model.  Furthermore, note that the error exponent of BICM is much larger than that of MLC (always being given by the minimum of the error exponents of the various levels, which results in an error exponent smaller than 1) \cite{beyer2001aaa,beyer2001aaa2,guillen2008bic,ingber2009cee}. Therefore, in terms of error probability, BICM outperforms MLC. 
\begin{figure}[t]
 \centering 
 \includegraphics[width=1\columnwidth]{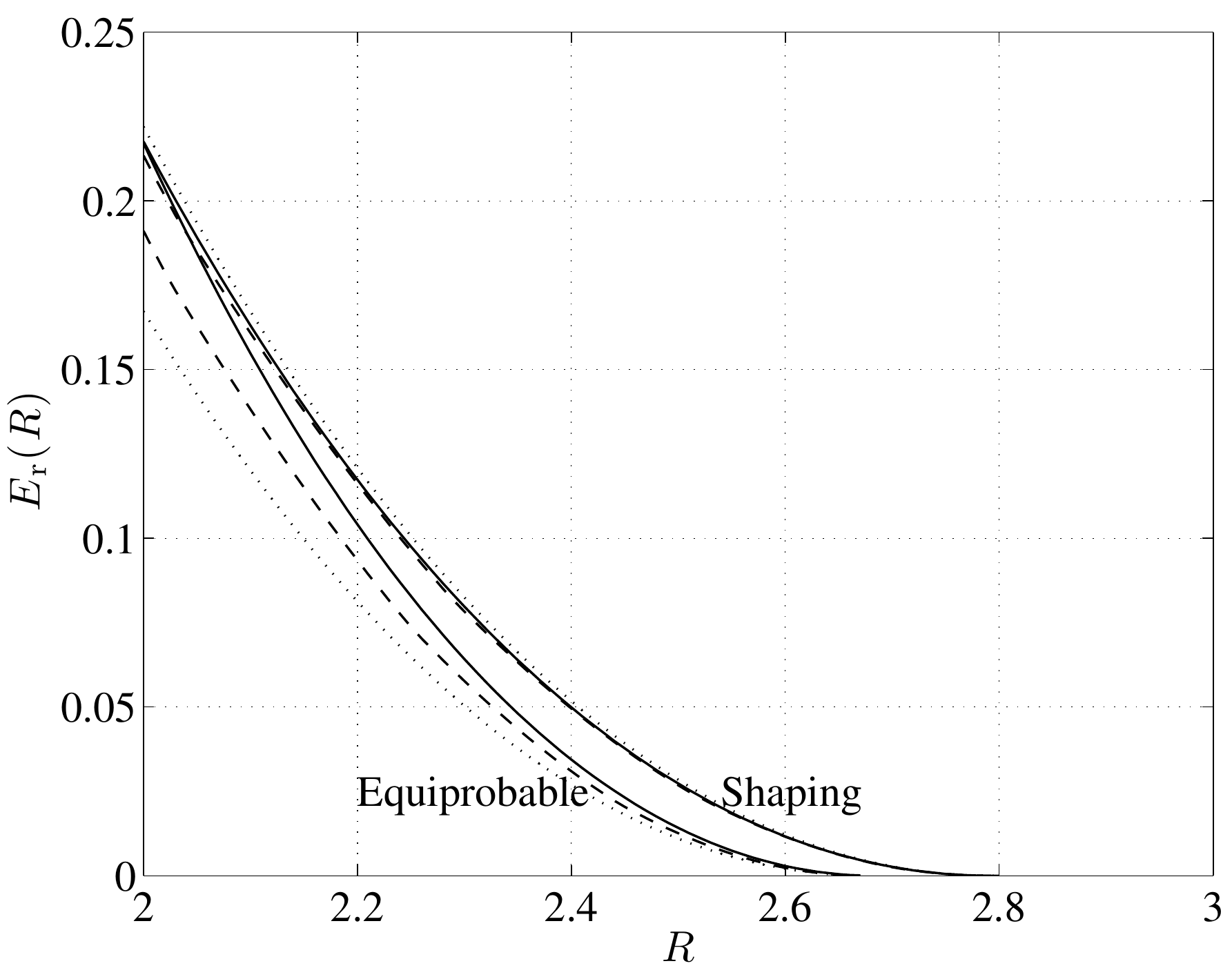}
 \caption{Error exponent zoom at capacity for 16-QAM with and without shaping for CM (solid) and BICM (dashed) at $\snr=8$ dB. The error exponent of the BICM parallel channel model is shown for comparison (dotted). When shaping is used, the input distribution is the corresponding optimal capacity achieving distribution.}
 \label{fig:exponents_cm_bicm}
   \vspace{-4mm}
\end{figure}

\subsection{Wideband Regime}

The gain from shaping in BICM is especially remarkable at low $\snr$, the wideband regime recently discussed at length by Verd{\'u} \cite{verdu2002sew}. Following his methodology, rather than studying the exact expression of the information rate, one considers a second-order Taylor series in terms of $\snr$,
\begin{align}
   R(\snr) = \ca\snr +\cb\snr^2 + \ord\bigl(\snr^2\bigr),
   \label{eq:c_snr}
\end{align}
where the
notation $\ord(\snr^2)$ indicates that the remaining terms vanish faster than a function $a\snr^2$, for $a > 0$ and small $\snr$. A scheme is  said to be first- and second-order optimal if $\ca = 1$ and $\cb = -\frac{1}{2}$, as it is for the channel capacity.
In those conditions, such a system is both power- and bandwidth-efficient. For instance, it is well known that for low $\snr$, QPSK is both first- and second-order efficient \cite{verdu2002sew}.

The low-$\snr$ performance of BICM was studied in \cite{martinez2008bic}, where general expressions for the coefficients $\ca$ and $\cb$ were given for general mapping rules and equiprobable signaling. For the particular case of binary reflected Gray mapping with squared QAM constellations, it was found that BICM was suboptimal, in the sense that it did not achieve the optimum $\ca$ and $\cb$.
References~\cite{stierstorfer2009asymptotically,alvarado2009ooc} propose alternative mapping rules for BICM that achieve $\ca=1$, or equivalently $\ebnolim\eqdef \frac{\log2}{\ca} = -1.59$ dB, with equiprobable signaling. Incidentally, this disproves the conjecture from \cite{caire1998bic} that binary-reflected Gray mapping is the optimum labeling rule for BICM schemes. However, the mappings of \cite{stierstorfer2009asymptotically,alvarado2009ooc} are not second-order optimal.

In the case of non-equiprobable signaling, binary reflected Gray mapping becomes optimal both in terms of $\ca$ and $\cb$. 
\begin{theorem}
Shaping makes BICM transmission over QAM modulations with binary reflected Gray mapping first- and second-order optimal, i.e., $\ca = 1$ and $\cb = -\frac{1}{2}$.
\end{theorem}

The key fact is that the bit probabilities are such that a QPSK constellation is effectively selected. To see how, note that for $m = 2$ we have QPSK with Gray mapping. Limiting ourselves to one dimension, the binary reflected Gray mapping for $\frac{m}{2} +1$ bits is constructed from the mapping for $\frac{m}{2}$ bits by prefixing a binary 0 to the mapping for $\frac{m}{2}-1$ bits, then prefixing a binary 1 to the reflected (i.\ e.\ listed in reverse order) mapping for $\frac{m}{2}-1$ bits. With shaping, one has the flexibility to fix each of this additional bits to a given value, say, 0, so that one is effectively transmitting over a BPSK constellation (QPSK over the two quadratures) when the resulting constellation is normalized in mean and energy. Note that this is property does not necessarily hold for other mapping rules.

\bibliographystyle{IEEE}



\end{document}